\documentclass{lmcs}

\usepackage{amsmath,amssymb,amsthm}
\usepackage{hyperref}
\usepackage{graphicx}
\usepackage{enumitem}
\usepackage{bookmark}

\theoremstyle{definition}
\newtheorem{definition}{Definition}[section]

\theoremstyle{plain}
\newtheorem{theorem}{Theorem}[section]
\newtheorem{proposition}{Proposition}[section]

\theoremstyle{remark}

\begin{document}

\title[Limited Math]{Limited Math: Aligning Mathematical Semantics with Finite Computation}
\author[L. Wen]{Lian Wen}

\address{School of ICT, Griffith University, Australia}
\email{l.wen@griffith.edu.au}

\begin{abstract}
Classical mathematical models used in the semantics of programming languages and
computation rely on idealized abstractions such as infinite-precision real numbers,
unbounded sets, and unrestricted computation. In contrast, concrete computation is
inherently finite, operating under bounded precision, bounded memory, and explicit
resource constraints. This discrepancy complicates semantic reasoning about numerical
behavior, algebraic properties, and termination under finite execution.

This paper introduces Limited Math (LM), a bounded semantic framework that aligns
mathematical reasoning with finite computation. Limited Math makes constraints on
numeric magnitude, numeric precision, and structural complexity explicit and
foundational. A finite numeric domain parameterized by a single bound \(M\) is equipped
with a deterministic value-mapping operator that enforces quantization and explicit
boundary behavior. Functions and operators retain their classical mathematical
interpretation and are mapped into the bounded domain only at a semantic boundary,
separating meaning from bounded evaluation.

Within representable bounds, LM coincides with classical arithmetic; when bounds are
exceeded, deviations are explicit, deterministic, and analyzable. By additionally
bounding set cardinality, LM prevents implicit infinitary behavior from re-entering
through structural constructions. As a consequence, computations realized under LM
induce finite-state semantic models, providing a principled foundation for reasoning
about arithmetic, structure, and execution in finite computational settings.
\end{abstract}

\maketitle

\section{Introduction}

Classical mathematical models form the foundation of programming language semantics
and formal reasoning about computation \cite{pierce2002tapl}. In these models, programs
and expressions are typically interpreted over idealized abstractions such as
infinite-precision real numbers, unbounded sets, and unrestricted computation. While
these abstractions enable elegant theory, they introduce a persistent semantic gap when
confronted with the inherently finite nature of concrete computation, which operates
under bounded precision, bounded memory, and explicit resource constraints.

This gap manifests in well-known and structurally significant ways. Finite-precision
arithmetic exhibits behaviors such as rounding anomalies, equality failures, and
non-associativity, complicating semantic reasoning about numerical expressions
\cite{goldberg1991floating,higham2002accuracy}. More generally, classical semantic
models often assume unbounded computational resources, even though concrete executions
necessarily range over finite representations. Even fundamental questions such as
termination are undecidable in classical models of computation
\cite{turing1936computable,sipser2012theory}, despite the fact that bounded machines
induce finite state spaces.

These discrepancies are not merely artifacts of implementation, but reflect a deeper
mismatch between idealized mathematical semantics and finite computation. Treating
finiteness as an implicit or secondary concern forces semantic reasoning to rely on
external assumptions about precision, magnitude, and structure, obscuring when and why
classical properties hold and when they fail.

In this paper, we introduce \emph{Limited Math (LM)}, a semantic framework designed to
align mathematical reasoning with finite computation. Rather than assuming unbounded
precision and structure, Limited Math makes constraints on numeric magnitude, numeric
precision, and structural complexity explicit and foundational. A single parameter
\(M\) determines the maximum representable magnitude, the minimum positive unit of
precision, and the total number of distinct representable numeric values. Arithmetic is
defined over a finite, uniformly discretized domain via a deterministic value-mapping
operator that enforces quantization and explicit boundary behavior.

Beyond numeric bounds, Limited Math also restricts the size of sets admitted by the
framework. By bounding set cardinality to scale with the representational capacity of
the numeric domain, LM prevents implicit infinitary behavior from re-entering the model
through combinatorial constructions. As a result, all collections are finite,
enumerable, and internally representable, ensuring that quantification and structural
reasoning remain compatible with bounded computation. This perspective is
philosophically aligned with finite model theory, while remaining focused on numerical
semantics rather than logical expressiveness \cite{libkin2004elements}.

A central design principle of Limited Math is the separation between semantic meaning and
bounded evaluation. Functions and operators retain their classical mathematical
interpretation, while a single semantic boundary maps results into the finite LM domain.
This approach avoids premature quantization, limits error accumulation in composed
expressions, and preserves classical algebraic behavior whenever intermediate results
remain within representable bounds. When bounds are exceeded, deviations are explicit,
deterministic, and semantically analyzable.

When Limited Math is realized under bounded memory, computation ranges over a finite
state space. Consequently, every execution either halts or eventually revisits a
previously encountered state and evolves cyclically. This observation does not alter
classical undecidability results, which rely on unbounded computation models, but makes
explicit the finite-state semantics induced by bounded execution.

Limited Math is not proposed as a replacement for classical mathematics or real
analysis. Instead, it provides a foundational semantic framework for reasoning about
finite-precision arithmetic, bounded structures, and finite computation. By making
resource bounds explicit rather than implicit, LM clarifies the semantic consequences
of finiteness and provides a principled basis for reasoning about computation in finite
mathematical settings.

\paragraph{Contributions.}
The main contributions of this work are:
\begin{itemize}[leftmargin=1.5em]
  \item a bounded numeric domain with explicit magnitude and precision constraints;
  \item a deterministic value-mapping operator aligning arithmetic semantics with finite computation;
  \item a principled restriction on set cardinality ensuring enumerability and internal representability;
  \item a semantic operator-mapping framework separating mathematical meaning from bounded evaluation;
  \item an analysis of algebraic properties within and beyond representable bounds;
  \item a finite-state semantic interpretation induced by bounded-memory computation.
\end{itemize}

\section{Design Principles of Limited Math}

Limited Math is guided by a small set of explicit design principles intended to reconcile mathematical rigor with practical computability. Rather than treating finiteness as an implementation artifact, LM incorporates resource bounds directly into its formal structure. This section outlines the core principles underlying the framework and explains their roles in shaping the theory.

\subsection{Bounded Numeric Magnitude}

The first principle of Limited Math is that all numeric values have bounded magnitude. A parameter \(M > 0\) determines the maximum representable absolute value. Any quantity whose real-valued magnitude exceeds this bound is deterministically mapped to the boundary value via a saturation rule.

Bounding magnitude reflects the reality of finite computation, where overflow is unavoidable and must be handled explicitly. In classical mathematics, magnitude is unbounded, allowing values to grow without restriction. In LM, by contrast, magnitude bounds are first-class and visible in the semantics. This makes overflow behavior explicit and analyzable rather than implicit or undefined.

\subsection{Bounded Numeric Precision}

The second principle is bounded precision. Limited Math represents numbers on a uniform discrete grid with minimum positive unit \(1/M\). All values are exact multiples of this unit, ensuring that every representable number has a finite description and can be manipulated without rounding ambiguity.

This choice aligns LM with fixed-point arithmetic rather than floating-point arithmetic. Unlike floating-point systems, LM avoids non-uniform spacing, hidden rounding modes, and exceptional values. Precision is explicit, uniform, and parameterized by \(M\), allowing users to choose an appropriate resolution for a given problem domain.

Bounding precision is not viewed as a limitation but as a deliberate modeling choice. It ensures that arithmetic operations are deterministic and that numerical error is bounded and predictable.

\subsection{Bounded Structural Complexity}

The third principle concerns structural complexity. In addition to bounding individual numeric values, Limited Math restricts the size of sets admitted by the theory. Specifically, all sets are required to have cardinality bounded by a function of \(M\), chosen to scale with the representational capacity of the numeric domain.

The motivation for this restriction is twofold. First, the numeric domain itself contains only finitely many distinct values. Admitting sets whose cardinality exceeds this capacity would introduce objects that cannot be indexed, enumerated, or distinguished internally. Second, many infinitary phenomena in classical mathematics arise not from numbers themselves, but from unrestricted set construction. By bounding set size, LM prevents implicit infinities from re-entering the theory through combinatorial means.

This restriction ensures that all collections are finite, enumerable, and internally representable. As a result, quantification over sets in LM is operationally meaningful and compatible with finite computation.

\subsection{Separation of Semantics and Execution}

A key design principle of Limited Math is the separation between semantic definition and execution. Mathematical functions and operators are interpreted at the level of classical real-valued mathematics, where their meaning is well established. Their results are then mapped into the bounded LM domain using a deterministic value-mapping operator.

This separation avoids premature quantization and minimizes error accumulation in composed operations. It also clarifies the status of analytic notions such as derivatives and limits, which are not internalized within LM but can be treated as semantic constructs whose outcomes are subsequently represented in the bounded domain.

By distinguishing meaning from representation, LM preserves conceptual clarity while maintaining strict execution constraints.

\subsection{Bounded Computation and Memory (Implementation Principle)}

When Limited Math is implemented on a computer, an additional practical constraint applies: memory is finite. Under a fixed memory bound, the operational semantics of any LM program induces a finite state space. Consequently, every execution either halts or eventually repeats a previously visited state.

This observation does not resolve classical undecidability results, which rely on unbounded computation models. Instead, it reflects the operational reality of real machines and provides a foundation for precise reasoning about program behavior in bounded environments. In this setting, termination becomes decidable in principle, as executions can be exhaustively explored within the finite state space.

This bounded-memory principle applies only to implementation and does not constrain the abstract mathematical framework beyond what is necessary for practical applicability.

\subsection{Summary}

Together, these principles define Limited Math as a bounded, explicit, and internally consistent framework. By constraining magnitude, precision, structural complexity, and—where relevant—computation, LM eliminates hidden infinities and aligns mathematical reasoning with finite execution. These design choices are not intended to replace classical mathematics, but to provide a controlled environment in which mathematical structures and algorithms can be studied under explicit resource bounds.

\section{The Limited Math Framework}

This section presents the formal definition of Limited Math. We begin by defining the bounded numeric domain, followed by the value-mapping operator that enforces magnitude and precision constraints. We then introduce the restriction on set cardinality that bounds
structural complexity.

\subsection{Numeric Domain}

Limited Math is parameterized by a positive integer \(b \ge 1\), representing the number of bits allocated to each numeric component. We define
\[
M := 2^{b} - 1.
\]

The parameter \(M\) determines both the maximum representable magnitude and the minimum positive unit of precision.

\begin{definition}[Limited Math Numeric Domain]
The Limited Math numeric domain is defined as
\[
\mathcal{N}_M
=
\left\{
\frac{k}{M}
\;\middle|\;
k \in \mathbb{Z},\;
-M^2 \le k \le M^2
\right\}.
\]
\end{definition}

Each value in \(\mathcal{N}_M\) is an exact multiple of \(1/M\), with maximum absolute value \(M\). The total number of distinct representable values is finite and on the order of \(M^2\). The maximum absolute value is $M$, corresponding to $k = \pm M^2$.

This domain corresponds naturally to a symmetric fixed-point representation with bounded magnitude and uniform precision. In the special case \(b=1\), we obtain \(M=1\) and \(\mathcal{N}_1 = \{-1,0,1\}\), which forms the minimal nontrivial Limited Math system.

\subsection{Value-Mapping Operator}

To relate real-valued quantities to the bounded domain \(\mathcal{N}_M\), Limited Math
employs a deterministic value-mapping operator that enforces both saturation and
quantization.

\begin{definition}[Value Mapping]
The value-mapping operator \(\Phi_M : \mathbb{R} \rightarrow \mathcal{N}_M\) is defined by
\[
\Phi_M(x) =
\begin{cases}
\;\;\;\; M, & x \ge M,\\[4pt]
- M, & x \le -M,\\[4pt]
\frac{\lfloor M x \rfloor}{M}, & -M < x < M.
\end{cases}
\]
\end{definition}

The operator \(\Phi_M\) maps every real number to a representable Limited Math value, with
quantization inside the representable range and saturation at the boundaries. This definition
ensures that the semantic interpretation of arithmetic in Limited Math is total and
deterministic.

\paragraph{Semantic versus execution behavior.}
The value-mapping operator \(\Phi_M\) is a \emph{semantic} construct: it defines how
unbounded mathematical values are interpreted within the bounded LM domain. When
Limited Math is embedded in a concrete programming language or execution model,
boundary behavior need not be realized by saturation at runtime. Instead, saturation
serves as a semantic characterization of boundary outcomes, while implementations may
choose to signal boundary violations explicitly (e.g., by trapping) in order to enforce
assumed numeric bounds.

\subsection{Arithmetic Operations}

Arithmetic operations in Limited Math are defined by composing classical real-valued
operations with the value-mapping operator.

\begin{definition}[LM Arithmetic]
For \(x,y \in \mathcal{N}_M\), define:
\[
x \oplus y := \Phi_M(x + y),
\qquad
x \otimes y := \Phi_M(x \cdot y).
\]
\end{definition}

These operations ensure closure of \(\mathcal{N}_M\) under addition and multiplication.
When intermediate real-valued results remain within the bounded range \([-M,M]\) and
align with the representable grid, LM arithmetic coincides exactly with classical arithmetic.
When boundaries are reached, the semantic outcome is characterized by \(\Phi_M\); concrete
implementations may instead detect such boundary crossings and signal them explicitly.

\subsection{Bounded Set Cardinality}

In addition to bounding numeric values, Limited Math restricts the size of sets admitted by the theory.

\begin{definition}[Bounded Sets]
All sets in Limited Math are required to have cardinality at most \(M^2\).
\end{definition}

This restriction reflects the finite representational capacity of the numeric domain \(\mathcal{N}_M\), which itself contains \(O(M^2)\) distinct values. Admitting sets larger than this bound would introduce objects that cannot be indexed or distinguished internally using LM arithmetic.

Bounding set cardinality ensures that all collections are finite, enumerable, and compatible with bounded computation. It also prevents implicit infinitary constructions from entering the framework through unrestricted set formation.

\subsection{Discussion}

Together, the numeric domain \(\mathcal{N}_M\), the value-mapping operator \(\Phi_M\), and the bounded set cardinality define the core structure of Limited Math. These components establish a finite yet scalable mathematical universe in which magnitude, precision, and structural complexity are explicit parameters.

In subsequent sections, we build on this framework to define function and operator mapping, analyze algebraic properties under bounded evaluation, and study the computational implications of Limited Math in practical settings.

\section{Operator and Function Mapping}

A central goal of Limited Math is to provide a rigorous correspondence between classical mathematical semantics and bounded execution. To this end, LM distinguishes between the semantic definition of functions and operators, which remains classical, and their execution within the bounded numeric domain. This section formalizes this separation and defines how functions and operators are mapped into Limited Math.

\subsection{Function Mapping}

Let \( f : \mathbb{R} \rightarrow \mathbb{R} \) be a real-valued function with a classical mathematical definition. Limited Math associates with \(f\) a bounded counterpart that operates on the LM numeric domain.

\begin{definition}[Mapped Function]
Given a real-valued function \( f : \mathbb{R} \rightarrow \mathbb{R} \), its Limited Math mapping is the function
\[
f^{(M)} : \mathcal{N}_M \rightarrow \mathcal{N}_M
\]
defined by
\[
f^{(M)}(x) := \Phi_M\!\left( f(x) \right),
\quad x \in \mathcal{N}_M.
\]
\end{definition}

This definition ensures that every function application produces a value within the bounded domain. The mapping is deterministic and total, regardless of the behavior of \(f\) outside the representable range.

This separation between semantic definition and execution aligns with standard approaches to programming language semantics, in which meaning is defined independently of representation or evaluation strategy \cite{pierce2002tapl}.

\subsection{Composition Semantics}

Function composition requires particular care in a bounded setting, as naive intermediate quantization can introduce unnecessary error accumulation and break algebraic structure.

Let \(f,g : \mathbb{R} \rightarrow \mathbb{R}\) be real-valued functions. Classical composition is defined by \((f \circ g)(x) = f(g(x))\). Limited Math defines the mapped composition as follows.

\begin{definition}[Mapped Composition]
The Limited Math composition of \(f\) and \(g\) is defined by
\[
(f \circ g)^{(M)}(x)
:=
\Phi_M\!\left( f(g(x)) \right),
\quad x \in \mathcal{N}_M.
\]
\end{definition}

Importantly, quantization is applied only once, after the full real-valued composition has been evaluated. Intermediate results are not snapped to the LM grid.

This \emph{snap-at-end} semantics minimizes quantization error and preserves semantic associativity:
\[
\bigl((f \circ g) \circ h\bigr)^{(M)}
=
\bigl(f \circ (g \circ h)\bigr)^{(M)}
=
\Phi_M\!\left( f(g(h(x))) \right).
\]

\subsection{Arithmetic Operators as Mapped Functions}

Arithmetic operators in Limited Math are special cases of the general mapping principle.

Let \(+\) and \(\cdot\) denote classical real-valued addition and multiplication. Their LM counterparts are defined by:
\[
x \oplus y := \Phi_M(x + y),
\qquad
x \otimes y := \Phi_M(x \cdot y),
\quad x,y \in \mathcal{N}_M.
\]

These definitions coincide with those introduced in Section~3 and are consistent with the function-mapping framework. Arithmetic operations are thus interpreted as classical operations followed by a single mapping into the bounded domain.

\subsection{In-Range Equivalence and Boundary Effects}

When intermediate real-valued results remain within the representable range and align with the LM grid, the mapping operator acts as the identity.

\begin{proposition}[In-Range Equivalence]
Let \(x,y \in \mathcal{N}_M\). If \(x + y \in \mathcal{N}_M\), then
\[
x \oplus y = x + y.
\]
Similarly, if \(x \cdot y \in \mathcal{N}_M\), then
\[
x \otimes y = x \cdot y.
\]
\end{proposition}

Outside this in-range regime, saturation and quantization introduce explicit deviations from classical arithmetic. These boundary effects are deterministic and analyzable, and they replace implicit overflow or rounding behavior commonly encountered in practical computation.

\subsection{Higher-Order Operators}

The mapping principle extends naturally to higher-order operators. Let
\[
\mathcal{O} : (\mathbb{R} \rightarrow \mathbb{R})^k \rightarrow (\mathbb{R} \rightarrow \mathbb{R})
\]
be a classical operator, such as differentiation or integration. Limited Math defines the mapped operator by:
\[
\mathcal{O}^{(M)}(f_1,\ldots,f_k)(x)
:=
\Phi_M\!\left(
\mathcal{O}(f_1,\ldots,f_k)(x)
\right),
\quad x \in \mathcal{N}_M.
\]

This approach preserves the classical semantic meaning of operators while enforcing bounded execution. Operators that rely on infinitary constructions are not internalized within LM; rather, their outcomes are mapped into the bounded domain when needed.

\subsection{Discussion}

The operator and function mapping framework provides a principled bridge between classical mathematics and bounded computation. By separating semantic definition from execution, Limited Math avoids premature quantization, preserves algebraic structure when possible, and makes deviations explicit when bounds are reached.

This design is particularly well suited to computational reasoning, where functions are specified abstractly but executed under finite precision and resource constraints.

This semantic separation mirrors the distinction between denotational meaning and operational realization commonly adopted in programming language design.

\section{Analysis of Algebraic Properties}

This section analyzes the algebraic behavior of Limited Math operations. We distinguish between an \emph{in-range regime}, where intermediate results remain within representable bounds, and a \emph{boundary regime}, where saturation or quantization occurs. This distinction clarifies which classical algebraic laws are preserved and which are necessarily broken in a bounded setting.

\subsection{In-Range Regime}

We say that an expression is evaluated \emph{in-range} if all intermediate real-valued results lie within the representable interval \([-M, M]\) and align with the LM numeric grid.

\begin{definition}[In-Range Evaluation]
An LM expression is said to be evaluated in-range if, during its evaluation under the snap-at-end semantics, all intermediate real-valued results belong to \(\mathcal{N}_M\).
\end{definition}

Within this regime, the value-mapping operator acts as the identity, and LM arithmetic coincides with classical arithmetic.

\begin{theorem}[Preservation of Addition Laws In-Range]
Let \(x,y,z \in \mathcal{N}_M\). If \(x+y\), \(y+z\), and \(x+y+z\) all belong to \(\mathcal{N}_M\), then:
\begin{itemize}[leftmargin=1.5em]
  \item \(x \oplus y = y \oplus x\) (commutativity),
  \item \((x \oplus y) \oplus z = x \oplus (y \oplus z)\) (associativity).
\end{itemize}
\end{theorem}

\begin{proof}
Under the stated conditions, no saturation or quantization occurs. By definition,
\[
x \oplus y = \Phi_M(x+y) = x+y,
\]
and similarly for all intermediate sums. Classical associativity and commutativity therefore apply directly.
\end{proof}

An analogous result holds for multiplication when products remain in-range and representable.

\subsection{Boundary Regime and Law Breaking}

When intermediate results exceed representable bounds or fall between grid points, the value-mapping operator introduces saturation or quantization. In this regime, certain classical algebraic laws no longer hold.

\begin{proposition}[Failure of Associativity under Saturation]
There exist \(x,y,z \in \mathcal{N}_M\) such that
\[
(x \oplus y) \oplus z \neq x \oplus (y \oplus z).
\]
\end{proposition}

\begin{proof}
Let \(x = M\), \(y = M\), and \(z = -M\). Then:
\[
(x \oplus y) \oplus z
=
\Phi_M(\Phi_M(2M) - M)
=
\Phi_M(M - M)
=
0,
\]
while
\[
x \oplus (y \oplus z)
=
\Phi_M(M + \Phi_M(0))
=
M.
\]
\end{proof}

This example demonstrates that associativity fails once saturation is involved. Importantly, the failure is explicit and deterministic.

\subsection{Distributivity}

Distributivity between addition and multiplication also holds only in the in-range regime.

\begin{proposition}[Conditional Distributivity]
Let \(x,y,z \in \mathcal{N}_M\). If \(x(y+z)\), \(xy\), and \(xz\) all lie in \(\mathcal{N}_M\), then
\[
x \otimes (y \oplus z) = (x \otimes y) \oplus (x \otimes z).
\]
\end{proposition}

Outside this regime, quantization or saturation may cause distributivity to fail. This behavior reflects the bounded nature of the numeric domain and replaces implicit overflow behavior with explicit semantics.

\subsection{Cancellation and Order Sensitivity}

Classical cancellation laws do not generally hold in Limited Math.

\begin{proposition}[Failure of Cancellation]
There exist \(x,y,z \in \mathcal{N}_M\) such that
\[
x \oplus y = x \oplus z
\quad \text{but} \quad
y \neq z.
\]
\end{proposition}

Such failures arise from saturation or quantization and imply that expression evaluation order can affect outcomes once boundary effects are present.

\subsection{Interpretation}

The loss of certain algebraic laws in the boundary regime is not a defect of Limited Math but a consequence of making bounds explicit. In classical mathematics, such boundary behavior is excluded by assumption; in practical computation, it is often hidden or undefined.

Limited Math exposes these effects directly and confines them to well-defined regions of computation. Within the in-range regime, classical reasoning applies unchanged. Outside it, deviations are explicit, deterministic, and analyzable.

\subsection{Summary}

Limited Math supports a two-tier algebraic interpretation:
\begin{itemize}[leftmargin=1.5em]
  \item \emph{In-range}, where LM arithmetic coincides with classical algebra;
  \item \emph{At boundaries}, where saturation and quantization introduce controlled deviations.
\end{itemize}

This structure enables precise reasoning about when classical algebraic laws may be safely applied and when bounded effects must be taken into account.

\section{Derivatives and Analytic Operators}

Analytic concepts such as limits, derivatives, and integrals play a central role in classical mathematics. However, these notions fundamentally rely on infinitary constructions and arbitrarily small perturbations. Since Limited Math operates over a finite, discretized domain with explicit bounds on precision and magnitude, such constructions cannot be internalized without violating the design principles of the framework.

This section clarifies the status of analytic operators in Limited Math and presents a principled approach for incorporating them without compromising rigor or stability.

\subsection{Why Derivatives Are Not Intrinsic to Limited Math}

In classical analysis, the derivative of a function \(f\) at a point \(x\) is defined as a limit:
\[
f'(x) = \lim_{h \to 0} \frac{f(x+h) - f(x)}{h}.
\]

This definition presupposes the existence of arbitrarily small nonzero values of \(h\) and an infinite process of refinement. In Limited Math, by contrast:
\begin{itemize}[leftmargin=1.5em]
  \item the numeric domain is finite and discrete;
  \item the minimum positive unit is fixed at \(1/M\);
  \item limits over infinite sequences are not expressible.
\end{itemize}

As a consequence, the classical derivative cannot be defined intrinsically within the LM numeric domain. Attempting to approximate derivatives using finite differences on the discretized grid leads to instability and artifacts, particularly when functions are quantized. For example, linear functions with small slopes may appear locally constant under grid-based differencing.

This limitation is not accidental but structural: any intrinsic derivative operator satisfying linearity, locality, and scale invariance would reintroduce infinitary assumptions incompatible with the bounded nature of LM.

\subsection{Semantic Interpretation of Analytic Operators}

Rather than internalizing analytic operators, Limited Math treats them as \emph{semantic constructs} defined at the level of classical mathematics. Their results are subsequently mapped into the bounded domain using the value-mapping operator.

\begin{definition}[Mapped Derivative]
Let \(f : \mathbb{R} \rightarrow \mathbb{R}\) be a differentiable function in the classical sense. The Limited Math derivative of \(f\) is defined by
\[
(Df)^{(M)}(x)
:=
\Phi_M\!\left( f'(x) \right),
\quad x \in \mathcal{N}_M.
\]
\end{definition}

This definition preserves the classical meaning of differentiation while ensuring that derivative values remain representable within the LM domain. Quantization is applied only once, at the final stage, avoiding instability due to repeated snapping.

\subsection{Stability Considerations}

The semantic mapping approach avoids common pathologies associated with discrete differentiation. For instance, consider the linear function \(f(x) = 0.3x\). Under naive finite-difference schemes applied directly to the LM-mapped function, small increments may be lost to quantization, yielding a derivative of zero over large regions. Under the mapped derivative definition, however, the classical derivative \(f'(x) = 0.3\) is preserved and represented exactly whenever \(0.3 \in \mathcal{N}_M\).

This illustrates a general principle: analytic meaning is preserved when differentiation is treated semantically rather than operationally within the bounded domain.

\subsection{Extension to Other Analytic Operators}

The same approach applies to other analytic operators, such as integration or higher-order differentiation. Let
\[
\mathcal{O} : (\mathbb{R} \rightarrow \mathbb{R})^k \rightarrow (\mathbb{R} \rightarrow \mathbb{R})
\]
be a classical analytic operator. Limited Math defines its bounded counterpart by
\[
\mathcal{O}^{(M)}(f_1,\ldots,f_k)(x)
=
\Phi_M\!\left(
\mathcal{O}(f_1,\ldots,f_k)(x)
\right),
\quad x \in \mathcal{N}_M.
\]

This unified treatment ensures consistency across analytic constructions while maintaining the boundedness and determinism of LM, without reintroducing infinitary assumptions into the numeric domain.

\section{Bounded Computation and Memory}

The preceding sections define Limited Math as a bounded mathematical framework at the level of values, functions, and structures. In this section, we consider the implications of these bounds when Limited Math is implemented on a computer. In particular, we examine the effect of imposing an explicit bound on memory and show how this leads to a finite-state execution model with strong theoretical guarantees.

In this section, we distinguish between execution \emph{semantics}, which describes the meaning of program behavior, and execution \emph{models}, which provide concrete mathematical structures realizing that behavior.

\subsection{Bounded-Memory Execution Model}

Real computational systems operate with finite memory. While classical models of computation, such as Turing machines, assume unbounded tapes, such assumptions are abstractions that do not hold in practice. When Limited Math is implemented on a computer, all data structures, registers, and control states are necessarily bounded.

We therefore consider an execution model in which:
\begin{itemize}[leftmargin=1.5em]
  \item numeric values range over the finite domain \(\mathcal{N}_M\);
  \item all sets have bounded cardinality;
  \item total memory usage is bounded by a fixed finite limit.
\end{itemize}

Under these assumptions, the global execution state of an LM program consists of a finite collection of components, including the program counter, bounded memory contents, and any auxiliary registers or stacks. The total number of such states is finite.

\subsection{Finite-State Execution Semantics}

Let \(\Sigma\) denote the set of all possible global execution states of a bounded-memory LM program. Since each component of the state is drawn from a finite set, \(\Sigma\) itself is finite.

Program execution induces a deterministic transition function
\[
\delta : \Sigma \rightarrow \Sigma \cup \{\mathsf{HALT}\},
\]
where \(\mathsf{HALT}\) denotes a terminal halting state.

Finite-state transition systems of this form are central to verification techniques such as model checking, where reachability and termination properties are analyzed over bounded state spaces \cite{clarke1999model}.

\subsection{Termination-or-Cycle Property}

The finiteness of the state space immediately yields the following result.

\begin{theorem}[Termination-or-Cycle]
For any deterministic Limited Math program executed under a fixed memory bound and any fixed input, the execution either halts after a finite number of steps or eventually revisits a previously encountered global state and thereafter evolves cyclically.
\end{theorem}

\begin{proof}
Since the set of global states \(\Sigma\) is finite, any infinite execution must visit some state more than once. Determinism of the transition function implies that once a state is revisited, the subsequent execution repeats the same sequence of states, forming a cycle. If no such repetition occurs, the execution must reach the halting state.
\end{proof}

\subsection{Decidability of Termination in the Bounded Setting}

An immediate consequence of the termination-or-cycle property is that termination is decidable within the bounded execution model.

Given a fixed LM program and input, one may simulate execution while recording visited states. If execution reaches the halting state, the program terminates. If a previously visited state is encountered, the execution will loop indefinitely and never halt.

This result does not contradict classical undecidability theorems, such as the Halting Problem, which apply to computation models with unbounded memory. Instead, it reflects the fact that bounded-memory systems correspond to finite-state machines, for which reachability and termination properties are decidable in principle.

\subsection{Practical Interpretation}

While the state space of a bounded-memory LM program may be extremely large in practice, the theoretical guarantee of termination-or-cycle behavior has important conceptual implications. It ensures that:
\begin{itemize}[leftmargin=1.5em]
  \item nontermination cannot arise from unbounded state growth;
  \item infinite execution corresponds to explicit cyclic behavior;
  \item program behavior is fully determined by a finite transition system.
\end{itemize}

These properties are particularly relevant to verification, safety-critical systems, and long-running computations, where understanding and controlling nontermination is essential.

\subsection{Summary}

By imposing explicit bounds on memory in addition to numeric and structural constraints, Limited Math yields finite-state execution semantics for implemented programs. Within this semantics, every computation either halts or eventually cycles, and termination becomes decidable in principle. This result reinforces the role of Limited Math as a framework that aligns mathematical reasoning with the realities of finite computation, without appealing to idealized unbounded models.

\section{Applications and Case Studies}

This section illustrates how Limited Math can be applied across a range of computational settings. Rather than aiming for exhaustive coverage, we focus on representative scenarios that highlight how explicit bounds on magnitude, precision, structure, and computation lead to clearer semantics and predictable behavior.

\subsection{Minimal System: Qualitative Reasoning with \texorpdfstring{$M=1$}{M=1}}

The smallest nontrivial Limited Math system arises when \(b=1\), yielding \(M=1\) and
\[
\mathcal{N}_1 = \{-1, 0, 1\}.
\]

Despite its simplicity, this system is useful for qualitative reasoning and decision-making. The three values may be interpreted as negative, neutral, and positive states, or as rejection, uncertainty, and acceptance. Arithmetic operations reduce to saturating addition and sign multiplication, yielding deterministic and interpretable behavior.

This minimal system is well suited to:
\begin{itemize}[leftmargin=1.5em]
  \item three-valued logic and rule evaluation;
  \item coarse-grained decision systems;
  \item control logic based on directional feedback;
  \item explainable reasoning in constrained environments.
\end{itemize}

The \(M=1\) case demonstrates that Limited Math is not merely a numerical approximation framework, but a family of systems that includes purely qualitative computation as a limiting case.

\subsection{Choosing \texorpdfstring{$M$}{M} by Problem Scale}

In practical applications, the parameter \(M\) is selected according to the scale and precision requirements of the problem domain. Smaller values of \(M\) yield coarse but robust computation, while larger values provide finer resolution at increased representational cost.

For example:
\begin{itemize}[leftmargin=1.5em]
  \item Small discrete systems (e.g., counters, team sizes, policy thresholds) may require only modest precision, making values such as \(M=15\) or \(M=31\) sufficient.
  \item General-purpose numerical computation can be supported with moderate values such as \(M=255\), corresponding to uniform fixed-point precision with resolution \(1/255\).
  \item Scientific and engineering applications often tolerate measurement uncertainty far larger than the minimum LM unit for \(M=1023\) or \(M=65535\), making such values adequate for many tasks.
\end{itemize}

This explicit parameterization contrasts with floating-point arithmetic, where precision and range are fixed by hardware and often opaque to users.

\subsection{Illustrative Instantiation: A C-Based Arithmetic Model}

Limited Math is intended as a semantic framework rather than a concrete programming
language design. Nevertheless, instantiating the framework within a familiar execution
model helps clarify how its semantic principles interact with finite representations and
bounded arithmetic. To this end, we consider an illustrative instantiation based on the
arithmetic model of the C programming language, whose treatment of fixed-width integers
closely reflects common hardware behavior. Notably, the C standard leaves signed integer
overflow undefined, while specifying other aspects of arithmetic behavior explicitly
\cite{iso9899}.

\paragraph{Fixed-point instantiation.}
We instantiate Limited Math using a signed fixed-point representation with 8 bits for
integer magnitude and 8 bits for fractional precision (Q8.8). Values are stored as
16-bit two’s-complement integers and interpreted as
\(
\mathsf{val}(x) = x / 2^8.
\)
This yields a finite numeric domain consistent with the LM framework, with uniform
precision and explicit bounds on magnitude.

\paragraph{Addition semantics.}
For two Q8.8 values \(x,y\), addition is performed by widening both operands to a
larger integer type, computing the sum, and checking whether the result lies within
the representable range. If so, the result is returned unchanged; otherwise, execution
\emph{traps}. Semantically, this corresponds to classical addition followed by a
single boundary check, mirroring the behavior of a hardware adder equipped with
overflow detection.

\paragraph{Multiplication semantics.}
Multiplication widens operands to produce a 32-bit intermediate product, which is
then rescaled by shifting right by 8 bits to restore Q8.8 precision. Optional symmetric
rounding may be applied prior to shifting. If the rescaled result lies outside the
representable range, execution \emph{traps}; otherwise, the result is returned as a
16-bit value. This reflects the standard widen--multiply--rescale datapath used in
fixed-point arithmetic.

\paragraph{Trap-based boundary policy.}
Unlike conventional fixed-point arithmetic, which often employs wraparound or
saturation, Limited Math adopts a \emph{trap-on-boundary} policy. Under this
interpretation, the bound \(M\) is assumed to be sufficient for all intended semantic
evaluations. Reaching the numeric boundary therefore signals a violation of the
assumed semantic domain rather than an alternative arithmetic outcome.

\paragraph{Semantic interpretation.}
Under this instantiation, LM arithmetic operators correspond to classical real-valued
operations followed by a single semantic boundary check. When all intermediate results
remain in-range, LM arithmetic coincides exactly with classical arithmetic. When a
boundary is exceeded, deviation is explicit and deterministic. This illustrates how
LM separates mathematical meaning from bounded evaluation while preserving classical
behavior within representable limits.

This example demonstrates that Limited Math can be instantiated within familiar
finite arithmetic models without introducing new representations or execution
mechanisms. Its purpose is illustrative: to show how bounded semantic interpretation
can be realized concretely while remaining independent of any particular programming
language design or implementation strategy.

\subsection{Bounded Computation and Verification}

When combined with bounded memory, Limited Math yields a finite-state execution model. This has direct implications for verification and analysis of programs.

In safety-critical and long-running systems, nontermination and uncontrolled state growth are major concerns. Under the LM execution model, every computation either halts or enters a cycle. This enables:
\begin{itemize}[leftmargin=1.5em]
  \item decidable termination analysis within bounded settings;
  \item explicit detection of cyclic behavior;
  \item exhaustive exploration of execution states in principle.
\end{itemize}

While the resulting state spaces may be large, the theoretical guarantees support precise reasoning about program behavior in constrained environments.

\subsection{Conceptual Alignment with Physical Bounds}

Although Limited Math is motivated by computation rather than physics, it is conceptually compatible with the notion that physical systems exhibit both minimum and maximum meaningful scales. For example, physical theories suggest a lower bound on spatial resolution and an upper bound on observable extent.

Interpreted cautiously, this perspective aligns with the bounded nature of LM: both reject unexamined infinities in favor of explicit limits. No claim is made that LM can simulate physical reality in full detail; rather, the analogy highlights the broader relevance of bounded mathematical reasoning.

\subsection{Summary}

These case studies demonstrate that Limited Math applies across a spectrum of computational contexts, from qualitative reasoning to numerical programming and bounded verification. By making bounds explicit and central to the theory, LM provides clearer semantics, predictable behavior, and a unified framework for reasoning about finite computation.

\section{Discussion, Related Work, and Scope}

Limited Math intersects with several established areas of research, including finite-precision arithmetic, bounded computation, and formal models of computation. This section situates LM within this landscape, clarifies its relationship to existing approaches, and discusses its intended scope and limitations.

\subsection{Relation to Floating-Point and Fixed-Point Arithmetic}

Finite-precision arithmetic has been studied extensively in both theory and practice,
spanning numerical analysis, programming language design, and hardware architecture.
Floating-point arithmetic, standardized by IEEE~754, is the dominant representation in
modern programming languages and hardware platforms \cite{ieee754, goldberg1991floating},
a design whose emphasis on making arithmetic behavior explicit as part of the programming
model dates back to the early development of floating-point standards \cite{kahan1996ieee}.

Despite its ubiquity, floating-point arithmetic is well known to exhibit unintuitive
behaviors arising from rounding, non-uniform spacing, non-associativity, and exceptional
values. These behaviors complicate formal reasoning about numerical correctness and have
motivated a substantial body of work on numerical stability, error analysis, and robust
algorithm design \cite{higham2002accuracy}. They also pose significant challenges for
program verification, particularly when floating-point computations are analyzed using
real-number abstractions \cite{monniaux2008floating}.

Reasoning about numeric behavior under bounded representations has long been addressed
through semantic abstraction techniques, most notably abstract interpretation
\cite{cousot1977abstract}. From a programming language perspective, the semantic gap
between idealized real-number models and finite-precision execution has long been
recognized as a source of subtle bugs and verification challenges. Language specifications typically describe arithmetic using
real-number abstractions while relying on underlying floating-point implementations, leaving
rounding behavior and boundary effects implicit or under-specified.
This mismatch complicates reasoning about properties such as equality, associativity, and
program equivalence in the presence of finite precision.

Fixed-point arithmetic provides a simpler and more predictable alternative by employing
uniform precision and explicitly bounded magnitude. It is widely used in embedded systems
and digital signal processing, where resource constraints and deterministic behavior are
paramount. However, fixed-point arithmetic is usually presented as an implementation or
representation technique rather than as a mathematical framework with explicit semantic
interpretation \cite{knuth1997taocp2}. Its correctness is typically argued at the level of
bit-width selection and scaling discipline, rather than through a formal semantic model.

Limited Math differs from both floating-point and fixed-point approaches by treating bounded
magnitude and bounded precision as axiomatic properties of the mathematical domain itself.
LM does not propose a new numeric format, nor does it prescribe a particular hardware or
language-level representation. Instead, it provides a semantic framework in which quantization,
boundary behavior, and deviation from classical arithmetic are explicit, deterministic, and
formally analyzable. In this sense, Limited Math complements existing numeric representations
by offering a principled basis for reasoning about their semantic consequences, rather than
serving as an alternative numeric encoding.

\subsection{Why Limited Math Is Not Just Fixed-Point Arithmetic}

A natural question is whether Limited Math reduces to fixed-point arithmetic under a
different presentation. While LM shares certain surface similarities with fixed-point
representations—such as uniform precision and bounded magnitude—the two differ
fundamentally in intent, abstraction level, and scope.

Classical discussions of fixed-point and seminumerical computation primarily focus on
representation choices, scaling strategies, and implementation techniques. These treatments
address how numbers are encoded and manipulated efficiently on finite hardware, as
exemplified by foundational work on seminumerical algorithms \cite{knuth1997taocp2}.
Correctness in this setting is typically argued operationally, in terms of avoiding overflow
or selecting sufficient bit-widths, rather than semantically.

Fixed-point arithmetic, as commonly used in practice, specifies how numeric values are
stored and computed, but it does not by itself provide a semantic account of how classical
mathematical meaning is related to bounded execution. In particular, it does not explain
how algebraic laws change at numeric boundaries, how error accumulates across composed
expressions, or under what conditions classical reasoning remains valid.

Limited Math, by contrast, is defined at the semantic level rather than the representational
level. Classical real-valued functions and operators retain their mathematical meaning, and
a single value-mapping operator mediates between unbounded semantics and bounded
execution. This separation of meaning from execution mirrors standard approaches in
programming language semantics, where denotational meaning is defined independently of
concrete evaluation strategies \cite{pierce2002tapl}.

By making bounded magnitude, bounded precision, and boundary behavior axiomatic rather
than incidental, LM provides a framework for reasoning about \emph{when} and \emph{why}
deviations from classical arithmetic occur. In this sense, Limited Math is not an alternative
numeric format, but a semantic framework for understanding bounded computation and its
consequences.

\subsection{Relation to Bounded Arithmetic and Finite Model Theory}

Bounded arithmetic and finite model theory investigate mathematical and logical systems
under explicit finiteness constraints, treating size bounds as fundamental rather than as
approximations of infinite structures. These areas have produced deep results concerning
expressiveness, decidability, and complexity by restricting the domains over which
computation and reasoning are performed.

Finite model theory, in particular, studies the behavior of logical formulas over finite
structures and has shown that finiteness fundamentally alters the properties of models,
expressibility, and definability \cite{libkin2004elements}. By rejecting implicit infinitary
assumptions, it emphasizes that reasoning over finite domains requires distinct semantic
and methodological tools.

Limited Math is philosophically aligned with this perspective in that it rejects unexamined
infinities and treats boundedness as a first-class property. The restriction on set cardinality
in LM prevents implicit combinatorial blow-up and ensures that all collections are finite,
enumerable, and internally representable. In this respect, LM shares the finite-model
intuition that unbounded constructions cannot be assumed without consequence.

However, Limited Math differs fundamentally in focus and intent. Bounded arithmetic and
finite model theory are primarily concerned with logical expressiveness, proof systems, and
complexity-theoretic classification. Limited Math, by contrast, is not a logical theory and
does not aim to characterize what can be proven or computed within a given complexity
class. Its focus is on numerical semantics: how bounded magnitude, bounded precision, and
bounded structure interact with arithmetic operations and program execution.

As a result, LM should be viewed as complementary rather than competing with these
theoretical frameworks. It adopts the finite perspective they champion, but applies it to
the semantic interpretation of numeric computation in programming languages, rather
than to logical definability or proof-theoretic strength.

\subsection{Relation to Classical Models of Computation}

Classical models of computation, including Turing machines and the lambda calculus,
assume unbounded memory and infinite domains. These abstractions underpin foundational
results such as the undecidability of the halting problem \cite{turing1936computable} and
remain indispensable to theoretical computer science \cite{sipser2012theory}. When
computation is restricted by explicit resource bounds, however, execution necessarily
ranges over a finite state space, a perspective that has been studied extensively in
formal verification through finite-state and resource-bounded automata models
\cite{alur1994timed}.

Limited Math does not challenge these models or their associated results. Instead, it follows
the standard semantic separation between abstract meaning and concrete execution models
emphasized in formal treatments of programming language semantics \cite{winskel1993semantics}.
It then adopts a complementary viewpoint that reflects the operational reality of computation
on physical machines. When numeric values, data structures, and memory are explicitly bounded,
program execution necessarily ranges over a finite state space.

Under such bounded conditions, execution can be modeled as a deterministic transition
system over a finite set of global states. In this setting, every computation either reaches
a halting state or eventually revisits a previously encountered state and evolves cyclically.
This observation is not a new theoretical result, nor does it weaken classical undecidability
theorems, which rely on unbounded models by design.

Rather, the contribution of Limited Math lies in making this finite-state interpretation
semantically explicit. By aligning mathematical reasoning with bounded execution, LM
provides a precise framework for understanding how classical notions of nontermination,
divergence, and undecidability manifest—or fail to manifest—in real computational systems.

\subsection{Scope and Limitations}

Limited Math intentionally restricts expressiveness in order to make resource bounds
explicit and semantically meaningful. In particular, analytic notions such as limits,
derivatives, and integrals are not intrinsic to the LM numeric domain and cannot be
defined operationally without reintroducing infinitary assumptions. Alternative approaches such as interval arithmetic provide correctness guarantees
by enclosing real-valued results within numeric bounds, but they pursue different
goals and trade-offs than the semantic alignment emphasized by Limited Math
\cite{moore2009introduction}.
 As discussed
earlier, such constructs are treated semantically rather than internalized within the
bounded framework.

Similarly, classical algebraic laws hold in Limited Math only under explicit in-range
conditions. Once numeric bounds or precision limits are reached, saturation or
quantization effects may cause associativity, distributivity, or cancellation laws to
fail. These deviations are not artifacts of the framework but direct consequences of
making finite precision and bounded magnitude explicit.

Limited Math therefore does not aim to replace classical mathematics, real analysis,
or unbounded models of computation. Its purpose is more modest and more precise:
to provide a semantic framework in which the consequences of finiteness are visible,
deterministic, and analyzable. LM is best suited to reasoning about numerical behavior,
program correctness, and execution properties in settings where bounded precision,
bounded structure, and bounded resources are inherent and unavoidable.

\subsection{Summary}

Limited Math draws on ideas from numerical computation, finite structures, and formal semantics while offering a distinct perspective. By treating boundedness as a foundational principle rather than an implementation detail, LM provides a coherent framework for reasoning about numerical behavior and computation in real programming languages. Its contribution lies in clarifying the semantic consequences of finiteness and making them explicit, analyzable, and mathematically rigorous.

\section{Conclusion and Future Work}

This paper introduced \emph{Limited Math (LM)}, a bounded mathematical framework designed to align formal reasoning with the realities of finite computation. By making bounds on numeric magnitude, numeric precision, structural complexity, and—when implemented—memory explicit, Limited Math eliminates implicit infinities that commonly arise in traditional mathematical abstractions used for computation.

At the core of LM is a finite numeric domain equipped with a deterministic value-mapping operator that enforces saturation and quantization. Functions and operators are interpreted semantically at the classical level and then mapped into the bounded domain, separating meaning from execution while minimizing unnecessary loss of precision. Within in-range conditions, LM coincides with classical arithmetic; at boundaries, deviations are explicit, deterministic, and analyzable. By additionally bounding set cardinality, LM ensures that all collections are enumerable and internally representable, preventing infinitary behavior from re-entering through combinatorial constructions.

When combined with bounded memory in computer implementations, Limited Math yields a finite-state execution model in which every computation either halts or eventually cycles. This result does not contradict classical undecidability theorems, but instead reflects the operational reality of finite machines and enables precise reasoning about program behavior within bounded environments.

Limited Math is not proposed as a replacement for classical mathematics or real analysis. Rather, it serves as a complementary framework for studying mathematics and computation under explicit resource constraints. Its value lies in providing clearer semantics, predictable behavior, and a unified foundation for reasoning about finite-precision arithmetic, bounded structures, and finite computation.

\paragraph{Future Work.}
Several directions for future research naturally follow from this work. These include deeper analysis of stability and error accumulation under LM arithmetic, integration of LM semantics into programming language specifications and verification tools, and exploration of LM-based models for optimization and control. Extensions to probabilistic or stochastic variants of Limited Math, as well as empirical comparisons with floating-point semantics in real-world programs, also represent promising avenues for further investigation.

By treating boundedness as a foundational principle rather than an implementation artifact, Limited Math opens a path toward a more explicit and rigorous understanding of computation as it is actually performed.

\bibliographystyle{plain}
\bibliography{reference.bib}

\end{document}